%% file: main.tex
\documentclass[letterpaper, 10 pt, conference]{ieeeconf}  

\IEEEoverridecommandlockouts                              %
\let\labelindent\relax                                                         

\usepackage{amsthm}
\newif\ifproof
\prooftrue
\usepackage{bm}
\usepackage{amsmath,amssymb}
\usepackage{xcolor}
\usepackage{subcaption}
\usepackage{booktabs} \usepackage{tikz}
\usetikzlibrary{automata,positioning}
\usetikzlibrary{shapes,arrows}
\tikzset{state/.style={circle, draw, minimum size=0.5cm, initial distance=0.2cm}}

\usepackage{todonotes}
\usepackage[noadjust]{cite}

\usepackage{algorithm}
\usepackage{enumitem}
\usepackage{algorithmicx}
\usepackage{algpseudocode}
\algnewcommand{\LineComment}[1]{\Statex \(\triangleright\) #1}
\makeatletter
\newcommand{\algmargin}{\the\ALG@thistlm}
\makeatletter
\newcommand*{\rom}[1]{\expandafter\@slowromancap\romannumeral #1@}
\makeatother
\newlength{\whilewidth}
\algdef{SE}[parWHILE]{parWhile}{EndparWhile}[1]
  {\parbox[t]{\dimexpr\linewidth-\algmargin}{     \hangindent\whilewidth\strut\algorithmicwhile\ #1\ \algorithmicdo\strut}}{\algorithmicend\ \algorithmicwhile}\algnewcommand{\parState}[1]{\State  \parbox[t]{\dimexpr\linewidth-\algmargin}{\strut #1\strut}}
\algnewcommand{\parRequire}[1]{\Require  \parbox[t]{\dimexpr\linewidth-\algmargin}{\strut #1\strut}}
\DeclareMathAlphabet{\mathpzc}{OT1}{pzc}{m}{it}
\usepackage[hidelinks]{hyperref}
\usepackage{cleveref}
\usepackage{etoolbox}
\usepackage{stackengine}
\def\delequal{\mathrel{\ensurestackMath{\stackon[1pt]{=}{\scriptstyle\Delta}}}}

\newtheorem{definition}{Definition}
\newtheorem{theorem}{Theorem}

\newtheorem{lemma}{Lemma}

\crefname{section}{Section}{Sections}
\crefname{subsection}{Section}{Sections}
\crefname{definition}{Definition}{Definitions}
\crefname{proposition}{Proposition}{Propositions}
\crefname{lemma}{Lemma}{Lemmas}
\crefname{theorem}{Theorem}{Theorems}
\crefname{corollary}{Corollary}{Corollaries}
\crefname{example}{Example}{Examples}
\crefname{figure}{Figure}{Figures}
\crefname{assumption}{Assumption}{Assumptions}
\crefname{remark}{Remark}{Remarks}
\crefname{running}{Running Example}{Running Examples}
\crefname{algorithm}{Algorithm}{Algorithms}

\newcommand\old[1]{{\color{gray} #1}}

\renewcommand\old[1]{}

\title{\LARGE \bf
Black-box Stealthy GPS Attacks on Unmanned Aerial Vehicles   
}

\author{Amir Khazraei, Haocheng Meng, and Miroslav Pajic
\thanks{The authors are with the Department of Electrical \& Computer Engineering, Duke University, Durham, NC 27708. Email: {\tt\small \{amir.khazraei, haocheng.meng, miroslav.pajic\}@duke.edu.}}
\thanks{This work is sponsored in part by the ONR under agreement N00014-23-1-2206, AFOSR under the award number FA9550-19-1-0169, and by the NSF CNS-1652544 award and the National AI Institute for Edge~Computing Leveraging Next Generation Wireless Networks, Grant CNS-2112562.}
}

\begin{document}

\maketitle
\thispagestyle{empty}
\pagestyle{empty}

\begin{abstract}
This work focuses on analyzing the vulnerability of unmanned aerial vehicles (UAVs) to stealthy black-box false data injection attacks on GPS measurements. We assume that the quadcopter is equipped with IMU and GPS sensors, and an arbitrary sensor fusion and controller are used to estimate and regulate the system's states, respectively.  We consider the notion of stealthiness in the most general form, where the attack is defined to be stealthy if it cannot be detected by any existing anomaly detector. Then, we show that if the closed-loop control system is incrementally exponentially stable, the attacker can cause arbitrarily large deviation in the position trajectory by compromising only the GPS measurements. We also show that to conduct such stealthy impact-full attack values, the attacker does not need to have access to the model of the system. Finally, we illustrate our results in a UAV case study.
\end{abstract}

\input{Intro}

\input{Preliminaries}
\input{Motive}

\input{Perfect}

\input{Simulation}

\input{Conclusion}

\bibliographystyle{IEEEtranMod}


\end{document}

%% file: Intro.tex
\section{Introduction}  
\label{sec:intro}

Unmanned aerial vehicles (UAVs), specifically quadcopters, have become indispensable tools across various applications such as aerial photography, search and rescue, military, agriculture, surveying and inspection, package delivery, and many more. Despite the immense potential, ensuring security and safety of UAVs remains a significant challenge that requires considerable attention. A critical sensor that UAVs heavily rely on is Global Positioning System (GPS) sensor, used for both localization and the navigation tasks. However, GPS sensors are known to be vulnerable to jamming and spoofing attacks~\cite{kerns2014unmanned}. Specifically, GPS spoofing provides control over the victim UAV without the need to  compromise the flight control software; where the attacker only needs to be equipped with a radio-frequency transmitter.

In particular, numerous methods have been proposed to show the impact of GPS attacks on UAVs~\cite{seo2015effect,sathaye2022experimental,he2018friendly,kerns2014unmanned}. There, the goal of the attacker has been to either deviate the drone's trajectory~\cite{seo2015effect}, or take over the control of the drone by navigating it through a desired trajectory~\cite{sathaye2022experimental,he2018friendly,kerns2014unmanned}. Since these methods have not considered stealthiness of the attacks, it has been shown that these attacks could be detected using either physics-based~\cite{liu2019analysis} or machine learning-based~\cite{brewington2023uav} anomaly detectors (AD). Hence, the moment an attack is detected, a recovery method and/or a safe planning/control algorithm could be applied to ensure safety (e.g., land the drone in a safe place). However, in case of a stealthy attack (i.e., that is undetectable by the AD), devising countermeasures becomes even more challenging as the usual initial step towards attack-resilience is attack detection.


The notion of attack stealthiness has been widely considered in attack-resilient control systems literature. For example, \cite{mo2010false,jovanov_tac19,kwon2014analysis,zhang2020false,khazraei_automatica21} study the impact of sensor attacks on linear time-invariant (LTI) control systems, aiming to cause large deviations in state estimates while remaining stealthy from a fixed AD. However, these works are limited to LTI systems and require the attacker to have access to the system model. Further, there is no guarantee that the attack will be stealthy against other types of ADs. Some recent works consider sensor attacks where stealthiness is independent of the deployed AD; i.e., the attack is stealthy if it is stealthy from any existing AD~\cite{bai2017kalman,zhang2022online}. For instance, \cite{bai2017kalman} examines sensor attacks on scalar LTI systems, while \cite{zhang2022online} extends the results to LTI systems with general state dimensions. However, they assume the attacker has full access to the system model, and the duration of the stealthy attack is~limited.

Since UAVs have nonlinear dynamical models, attacks relying on LTI models are less effective. Vulnerability of nonlinear control systems to stealthy sensor attacks, independent of the deployed AD, has been recently addressed~\cite{khazraei2022attacks,khazraei2023vulnerability,khazraei2023resiliency_cdc}. It has been shown that if the attacker has access to the system states and model, and the system meets certain properties, it is vulnerable to stealthy impactful attacks.

This work aligns with prior research by considering sensor attacks on nonlinear control systems~\cite{khazraei2022attacks, hallyburton_cdc22, khazraei2023vulnerability, khazraei2023resiliency_cdc,khazraei2023stealthy}. Yet, there are two main distinctions: the sensor attacks considered here are \emph{black-box}, meaning the attacker does not need access to the system model or states, and we show that a stealthy, impactful attack can be launched solely by compromising the GPS sensor. We propose attack vectors that, by only compromising the GPS sensor, cause deviation in the drone's trajectory while remaining stealthy from \emph{any} existing AD. The contributions of this paper are twofold: our attack is black-box, not requiring access to the system model or states, and we provide conditions under which the attack remains stealthy from any AD. Specifically, we show if the closed-loop control system is incrementally stable, the attack will remain stealthy while deviating the UAV's trajectory.


\paragraph*{Notation}
We use $\mathbb{R, Z}, \mathbb{Z}_{\geq 0}$ to denote  the sets of reals, integers and non-negative integers, respectively.  For 
a square matrix $A$, $\lambda_{max}(A)$ denotes the maximum eigenvalue. 
For a vector $x\in{\mathbb{R}^n}$, $||x||$ denotes the $2$-norm. 
For a vector sequence, 
$x_{0:t}$ denotes the set $\{x_0,x_1,...,x_t\}$. $\mathbf{0}_n$ is a vector of zeros with dimension $n$. We use $\textbf{mod} (a,b)$ to denote the remainder of $a$ divided by $b$, for $a,b\in \mathbb{Z}$. For any sets $A$ and $B$, $A-B$ denotes the set that contains all the elements that are in $A$ but not in $B$.
%
Finally, if $\mathbf{P}$ and $\mathbf{Q}$ are probability distributions relative to Lebesgue measure with densities $\mathbf{p}$ and $\mathbf{q}$, respectively, then 
the Kullback–Leibler  (KL) divergence between $\mathbf{P}$ and $\mathbf{Q}$ is
$KL(\mathbf{P}||\mathbf{Q})=\int \mathbf{p}(x)\log{\frac{\mathbf{p}(x)}{\mathbf{q}(x)}}dx$.


%% file: Preliminaries.tex
\section{Preliminaries}\label{sec:prelim}

Let $\mathbb{X}\subseteq \mathbb{R}^n$ and $\mathbb{D}\subseteq \mathbb{R}^m$. Consider a discrete-time nonlinear system with an exogenous input, modeled in the state-space form as
\begin{equation}\label{eq:prilim}
x_{t+1}=f(x_t,d_t,t),\quad x_t\in \mathbb{X},\,\,t\in \mathbb{Z}_{\geq 0},
\end{equation}
where $f:\mathbb{X}\times\mathbb{D}\times  \mathbb{Z}_{\geq 0}\to \mathbb{X}$ is continuous. We denote by $x(t,\xi,d)$ 
the trajectory (i.e., the solution) of~\eqref{eq:prilim} at time $t$, when the system has the initial condition $\xi$ and is subject to the input sequence $d_{0:t-1}$.\footnote{To simplify our notation, we denote the sequence $d_{0:t-1}$ as $d$.}  The following definition is derived from~\cite{angeli2002lyapunov,tran2018convergence,tran2016incremental}.
\begin{definition}
The system~\eqref{eq:prilim} is incrementally exponentially stable (IES) in the set $\mathbb{X}\subseteq \mathbb{R}^n$ if exist $\kappa>1$ and
$\lambda>1$~that
\begin{equation}
\Vert x(t,\xi_1,d)-x(t,\xi_2,d)\Vert \leq \kappa \Vert \xi_1-\xi_2\Vert \lambda^{-t},
\end{equation}
holds for all $\xi_1,\xi_2\in \mathbb{X}$, any $d_t\in \mathbb{D}$, and $t\in \mathbb{Z}_{\geq 0}$. When $\mathbb{X}=\mathbb{R}^n$, the system is referred to as globally incrementally exponentially stable (GIES).
\end{definition}

We now present some properties of KL divergence known as monotonicity and chain-rule~\cite{polyanskiy2022information}.

\begin{lemma}{\cite{polyanskiy2022information}}\label{lemma:mon}
\textbf{(Monotonicity):} Let $P_{X,Y}$ and $Q_{X,Y}$ be two distributions for a pair of variables $X$ and $Y$, and $P_{X}$ and $Q_{X}$ be two distributions for variable $X$. Then,
\begin{equation*}
KL(Q_X||P_X)\leq KL(Q_{X,Y}||P_{X,Y})
\end{equation*}
\end{lemma}

\begin{lemma}{\cite{polyanskiy2022information}}\label{lemma:chain}
\textbf{(Chain rule):} Let $P_{X,Y}$ and $Q_{X,Y}$ be two distributions for a pair of variables $X$ and $Y$. Then,
\begin{equation*}
KL(Q_{X,Y}||P_{X,Y})= KL(Q_{X}||P_{X})+KL(Q_{Y|X}||P_{Y|X}),
\end{equation*}
where $KL(Q_{Y|X}||P_{Y|X})$ is defined as
\begin{equation*}
KL(Q_{Y|X}||P_{Y|X})=\mathbb{E}_{x\sim Q_X}\{KL(Q_{Y|X=x}||P_{Y|X=x})\}.
\end{equation*}
\end{lemma}

\begin{lemma}{\cite{polyanskiy2022information}}\label{lemma:Guassian}
Let $P_{X}$ and $Q_{X}$ be two Gaussian distributions with the same covariance $\Sigma$ and different means of $\mu_P$ and $\mu_Q$, respectively. Then, it holds that 
\begin{equation*}
KL(Q_{X}||P_{X})= \mu_Q^T\Sigma^{-1} \mu_P.
\end{equation*}
\end{lemma}

%% file: Motive.tex
\section{System Model} 
\label{sec:motive} 

\begin{figure}[!t]
\centering
\includegraphics[width=.7\linewidth, height=3.64cm]{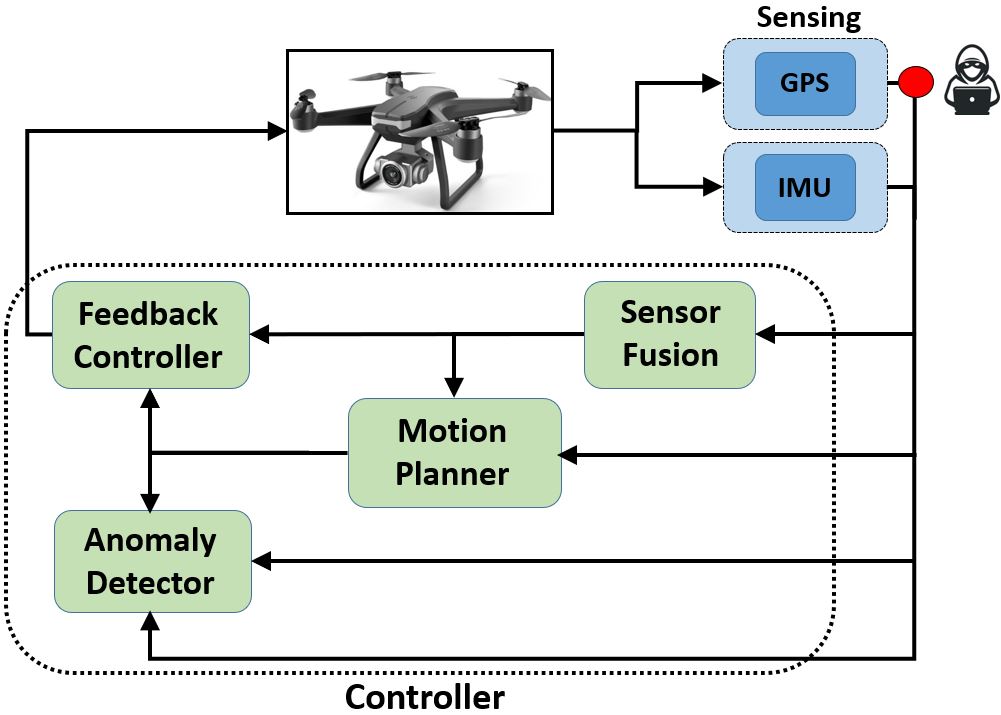}
\caption{Control system architecture in the presence of attacks.}
\label{fig:architecture}
\end{figure}

This section introduces the system and attack model,  allowing us to formally capture the problem addressed in this work. 

\subsection{System and Attack Model} \label{sec:A}
We consider the setup from \cref{fig:architecture} where each component is modeled as follows.

\subsubsection{UAV's Dynamical Model}
The quadcopter's dynamics are described by the following Newton-Euler equations~\cite{bouabdallah2007full}
\begin{equation}\label{eq:UAV_state_sp}
\begin{split}
\ddot{\phi}=&(\frac{I_y-I_z}{I_x})\dot{\theta}\dot{\psi}+\frac{\tau_x}{I_x}, \,\,\, \ddot{\theta}=(\frac{I_z-I_x}{I_y})\dot{\phi}\dot{\psi}+\frac{\tau_y}{I_y}\\
\ddot{\psi}=&(\frac{I_x-I_y}{I_z})\dot{\theta}\dot{\phi}+\frac{\tau_z}{I_z},  \,\,\, \ddot{z}= g-\frac{F}{m}\bigg(\cos(\phi) \cos(\theta) \bigg)\\
\ddot{x}=& \frac{F}{m}\bigg(\cos(\phi) \sin(\theta) \cos(\psi) + \sin(\phi) \sin(\psi)\bigg)\\
\ddot{y}=& \frac{F}{m}\bigg(\cos(\phi) \sin(\theta) \sin(\psi) - \sin(\phi) \cos(\psi)\bigg);
\end{split}
\end{equation}
here, $\mathbf{p}=\begin{bmatrix}x,y,z\end{bmatrix}^T$ and $\mathbf{v}=\begin{bmatrix}\dot{x},\dot{y},\dot{z}\end{bmatrix}^T$ represent the position and velocity in the global frame $\mathcal{E}$, while $\dot{\Omega} = \begin{bmatrix}\dot{\phi},\dot{\theta},\dot{\psi}
\end{bmatrix}^T$ and $\dot{\Omega} = \begin{bmatrix}\dot{\phi}, \dot{\theta}, \dot{\psi}\end{bmatrix}^T$ denote roll, pitch and yaw angles and their rates. The vehicle mass is $m$, $g$ is the gravitational acce- leration, $F$ is the total thrust, $\bm{\mathcal{I}} = \text{diag}(I_x, I_y, I_z)\in \mathbb{R}^{3\times 3}$ is the inertia matrix, and $\bm{\tau}=\begin{bmatrix}\tau_x,\tau_y,\tau_z\end{bmatrix}^T$ is the torque~vector.

Further, the control inputs are the squared angular velocities of the rotors, $\mathbf{u}=\begin{bmatrix}w_1^2,w_2^2,w_3^2,w_4^2\end{bmatrix}^T$. The thrust and torques are calculated as
%
\begin{equation}
\begin{bmatrix}F\\\tau_x\\\tau_y\\\tau_z\end{bmatrix}=\begin{bmatrix}b & b & b & b\\0 & -bl & 0 &bl\\bl & 0 & -bl &0\\-d & d & -d & d \end{bmatrix}\begin{bmatrix}w_1^2\\w_2^2\\w_3^2\\w_4^2\end{bmatrix},
\end{equation}
where $l$ is the distance from the motor to the center of gravity, and $b$ and $d$ are thrust and drag coefficients, respectively. The state vector is $\mathbf{x}=\begin{bmatrix} \phi ,\, \theta ,\, \psi ,\, \dot{\phi} ,\, \dot{\theta} ,\, \dot{\psi} ,\, x ,\, y ,\, z ,\, \dot{x} ,\, \dot{y},\, \dot{z} \end{bmatrix}^T$, leading to the state-space representation
dynamics
\begin{equation}\label{eq:dynamic_cont}
\dot{\mathbf{x}}=\mathbf{f}(\mathbf{x}_t,\mathbf{u}_{t_k}), \qquad t\in [t_k
,t_k+T_d) ;
\end{equation}
here, $t_k = kT_d$ is the sampling time ($T_d$ is the controller's sampling period), and $\mathbf{f}$ the 
transition function from~\eqref{eq:UAV_state_sp}. 
For 
our analysis, we discretize the system using the Euler~method
%
\begin{equation}\label{eq:discrete_1}
{\mathbf{x}}_{t_k+T_d}={\mathbf{x}}_{t_k} + \mathbf{f}(\mathbf{x}_{t_k},\mathbf{u}_{t_k})T_d.
\end{equation}
To simplify notation we drop $T_d$, and including model uncertainty, 
\eqref{eq:discrete_1} becomes
\begin{equation}\label{eq:discrete_2}
{\mathbf{x}}_{k+1}={\mathbf{x}}_{k} + \mathbf{f}(\mathbf{x}_{k},\mathbf{u}_{k})T_d+\mathbf{w}_k, 
\end{equation}
where $\mathbf{w}$ is a Gaussian white noise with zero mean and covariance $\Sigma_{\mathbf{w}}$. We assume that the system is equipped with IMU and GPS sensors that can imperfectly measure a subset of states with sampling intervals of $T_{imu}$ and $T_{gps}$, respectively. The GPS sampling frequency is typically $5 \text{Hz}$ ($T_{gps} = 200 \text{ms}$), while the IMU operates at $1 \text{kHz}$ ($T_{imu} = 1 \text{ms}$). Since the controller and sensor fusion modules run at the same frequency as the IMU, we consider $T_d = T_{imu} = 1 \text{ms}$.  
Defining $N_{g}=\frac{T_{gps}}{T_{imu}}$, the GPS updates at $k\in \mathbb{Z}^g$, where $\mathbb{Z}^g = \{k\in \mathbb{Z}|\textbf{mod}(k,N_{g})=0\}$. 

The GPS directly measures the drone's position
\begin{equation}
 \mathbf{y}_k^{\mathbf{p}}= \mathbf{p}_k+\mathbf{n}_k^{\mathbf{p}},
\end{equation}
$k\in \mathbb{Z}^g$, with $\mathbf{n}_k^{\mathbf{n}}$ zero mean Gaussian noise with covariance $\Sigma_{\mathbf{p}}$ taken into account. By augmenting Euler angles and their rates as $\mathbf{\Omega}=\begin{bmatrix} \Omega^T & \dot{\Omega}^T\end{bmatrix}^T$,  we model the IMU sensor as
\begin{equation}
 \mathbf{y}_k^{\mathbf{\Omega}}= h(\mathbf{\Omega}_k)+\mathbf{n}_k^{\mathbf{\Omega}} 
\end{equation}
for all $k\in \mathbb{Z}$. Moreover, $h$ is a Lipschitz nonlinear function\footnote{For more information about the function $h$ please refer to~\cite{mahony2012multirotor}.} with constant $L_h$ and  $\mathbf{n}_k^{\mathbf{\Omega}}$ is a Gaussian noise with mean ${\mu_{\mathbf{\Omega}}}_k$,\footnote{Note that the mean can be nonzero to capture the IMU bias.} and covariance $\Sigma_{\mathbf{\Omega}}$.  We also define $\mathbf{y}_k=\begin{bmatrix}
  \mathbf{y}_k^{\mathbf{p}} \\ \mathbf{y}_k^{\mathbf{\Omega}}   
\end{bmatrix}$ for $k\in \mathbb{Z}^g$ and $\mathbf{y}_k=\mathbf{y}_k^{\mathbf{\Omega}}$ for $k\in \mathbb{Z}-\mathbb{Z}^g$, capturing all sensor measurements available at time $k$.


\subsubsection{Control Unit}
The controller, shown in~\cref{fig:architecture}, includes a sensor fusion module (state estimator), a feedback controller, and an anomaly detector (AD). Below, we provide details on the sensor fusion and feedback controller. The AD will be discussed after introducing the attack model.

\paragraph*{Sensor Fusion}
Accurate state estimation is essential for effective control of the UAV's attitude and position.~As~IMU measurements do not directly measure $\mathbf{\Omega}$ and both~the~IMU and GPS are subject to noise, 
sensor fusion is necessary to estimate the states. Considering that the IMU has a higher sampling rate than the GPS ($T_{imu}<T_{gps}$), we update the state estimates based on the IMs sampling rate. In a general form, state estimates evolve using the  dynamics
\begin{equation}\label{eq:fusion}
\hat{\mathbf{x}}_{k} = \mathbf{g}(\hat{\mathbf{x}}_{k-1},\mathbf{u}_{k-1},\mathbf{y}_k^{c},k),
\end{equation}
where $k$ in the function $\mathbf{g}$ captures the time dependency of the sensor fusion, as $\mathbf{y}_k^{c}$ contains GPS sensor output only for $k\in \mathbb{Z}^g$. 
$\hat{\mathbf{x}} \in \mathbb{R}^{n}$ is a vector containing the state estimations, and $\mathbf{g}$ is a Lipschitz function with constant $L_{\mathbf{g}}$ (in 
sensor fusion models, it is commonly $n > 12$~\cite{angeli2002lyapunov}). Also, $\mathbf{y}_k^{c}$ captures the sensor measurements received by the controller. Thus, without malicious activity, it holds that $\mathbf{y}_k^{c}=\mathbf{y}_k$.

\paragraph*{Controller}
Since UAVs are intrinsically unstable, stabilizing them with an appropriate controller is critical. 
A feedback controller can be represented in the state-space~form~as
\begin{equation}
\label{eq:plant_withoutB}
\mathpzc{X}_{k}=f_c(\mathpzc{X}_{k-1},\hat{\mathbf{x}}_{k},r_k), \qquad
\mathbf{u}_k  = K(\mathpzc{X}_{k},\hat{\mathbf{x}}_{k});
\end{equation}
here, $\mathpzc{X}_k$ is the internal state of the controller, and $r_k$ represents the set point values provided by the planning module. The function $f_c$ evolves the controller's internal state and is assumed to be Lipschitz with constant $L_{f_c}$. The function $K$ maps $\mathpzc{X}$ and $\hat{\mathbf{x}}$ to the control input.
The full state of the closed-loop control system is defined as $\mathbf{X}\delequal \begin{bmatrix}
\mathbf{x}^T&\hat{\mathbf{x}}^T&{\mathpzc{X}}^T \end{bmatrix}^T$, and exogenous disturbances and set points are augmented as  $\mathbf{d}\delequal \begin{bmatrix}
\mathbf{w}^T&\mathbf{n}^T&r^T \end{bmatrix}^T$. The dynamics of the closed-loop system can then be represented~as
\begin{equation}\label{eq:closed-loop}
\mathbf{X}_{k+1}=F(\mathbf{X}_k,\mathbf{d}_k,k).
\end{equation}

\subsubsection{Attack Model} \label{sec:attack_model}
We consider a sensor attack model targeting GPS measurements, where the information received by the controller differs from the actual measurements. Such attacks can be non-invasive (e.g., GPS spoofing~\cite{kerns2014unmanned}), or involve compromising the information flow between the GPS sensor and the controller (e.g., network-based attacks~\cite{lesi_rtss17}). In either case, the attacker can `inject' a desired false~data into the current GPS sensor measurements.
Thus, assuming the attack begins at time $k=0$, the compromised GPS measurements delivered to the controller 
can be modeled~as~\cite{ncs_attack_model}
\begin{equation}\label{att:model}
\mathbf{y}_k^{\mathbf{p},c,a} = \mathbf{y}_k^{\mathbf{p},a}+a_k,~~k\in \mathbb{Z}_{\geq 0}^g;
\end{equation}
here, $a_k\in {\mathbb{R}^3}$ represents the attack signal injected at time $k$, and $\mathbf{y}_k^{\mathbf{p},a}$ is the true sensor information (i.e., without the attack) 
at time $t$. Since the sensor fusion and controller use the compromised sensor data to compute state estimates $\hat{\mathbf{x}}_t$ and input $\mathbf{u}_t$, the attack impacts the system's evolution. We use the superscript $a$ to denote any signal from a compromised system, e.g., $\mathbf{y}_k^{\mathbf{p},a}$ for pre-attack measurements, and $\mathbf{X}^a\delequal \begin{bmatrix}
{\mathbf{x}^a}^T & \left(\hat{\mathbf{x}}^a\right)^T&{\mathpzc{X}^a}^T \end{bmatrix}^T$ for the closed-loop plant state under attack.

In this study, we assume that the attacker does \emph{not} have access to the system model, sensor fusion, or controller dynamics -- i.e., we consider \emph{black-box attacks}. However, the attacker has computation power to compute suitable attack signals. The objective is to design an attack signal $a_k$, $k \in \mathbb{Z}_{\geq 0}$, that remains \emph{stealthy} -- i.e., undetected by the AD -- while \emph{causing deviations} in the UAV trajectory. The formal definition of the attacker's goal will be provided after introducing the AD and the notion of stealthiness.

\subsubsection{Anomaly Detector} 
To detect system attacks and anomalies, an AD analyzes received IMU and GPS sensor measurements. The AD accesses a sequence of values $\mathbf{y}_{-\infty:k}^c$ and performs binary hypothesis checking

$H_0$:  normal condition (the AD receives $\mathbf{y}_{-\infty:k}^c$);~~

$H_1$: abnormal behaviour (receives 
$\mathbf{y}_{-\infty:-1}^c,\mathbf{y}_{0:k}^{c,a}$).\footnote{Since the attack starts at $t=0$, we do not use superscript $a$ for the system evolution for $t<0$, as the trajectories of the non-compromised and compromised systems do not differ before the attack~starts.}

Given a sequence of received measurements $\bar{\mathbf{y}}^k=\bar{\mathbf{y}}_{-\infty:k}$, it is either from the $H_0$ hypothesis with the joint distribution $\mathbf{P}(\mathbf{y}_{-\infty:k}^c)$, or from the alternate hypothesis with a joint distribution $\mathbf{Q}(\mathbf{y}_{-\infty:-1}^c,\mathbf{y}_{0:k}^{c,a})$  -- note that this distribution is not known.
We define the anomaly detector $D$ as the mapping
\begin{equation}
D: \bar{\mathbf{y}}^k \to \{0,1\},
\end{equation}
where the output $0$ corresponds to $H_0$ and output $1$ to $H_1$. For any anomaly detector $\mathcal{D}$ we define probability of true detection as $p^{TD}=\mathbb{P}(D(\bar{\mathbf{y}}^k)=1|\bar{\mathbf{y}}^k \sim \mathbf{Q})$ and probability of false alarm as $p^{FA}=\mathbb{P}(D(\bar{\mathbf{y}}^k)=1|\bar{\mathbf{y}}^k \sim \mathbf{P})$.




\section{Formalizing Stealthiness and Attack Goals}
\label{sec:stealthy}
In this section, we capture the conditions for which an attack sequence is stealthy from \emph{{any}} 
AD.  Specifically, an attack is defined to be strictly stealthy if there exists no detector that can detect the  attack better than a random guess detector; here, by better we assume that the probability of true detection is greater than the probability of false alarm. However, achieving such stealthiness may not always be feasible. Therefore, we introduce $\epsilon$-\emph{stealthiness}, which, as demonstrated later, can be achieved for GPS attacks. For a detailed discussion on attack stealthiness see~\cite{khazraei2022attacks}. 

\begin{definition} \label{def:stealthiness}
Consider the system from~\eqref{eq:closed-loop}. An attack sequence, 
$\{a_{0}, a_{1},...\}$, 
is \emph{\textbf{strictly stealthy}} if there exists no detector for which $p_t^{TD}-p_t^{FA}>0$ holds, for any $t\geq 0$.  
An attack is
\textbf{$\epsilon$-\emph{stealthy}} if for a given $\epsilon >0$, there exists no detector such that $p_t^{TD}-p_t^{FA}>\epsilon$ holds, for any $t\geq 0$. 
\end{definition}

The following Lemma uses Neyman-Pearson theorem to capture the condition for which the received sensor measurements 
satisfy the stealthiness condition from Definition~\ref{def:stealthiness}. 

\begin{lemma}[\cite{khazraei2022attacks,khazraei_l4dc22}]
An attack sequence 
is $\epsilon$-stealthy if the corresponding observation sequence $\mathbf{y}_{0}^{c,a}:\mathbf{y}_t^{c,a}$ satisfies
\begin{equation}\label{ineq:stealthiness}
KL\big(\mathbf{Q}(\mathbf{y}_{-\infty:-1}^c,\mathbf{y}_{0:k}^{c,a})||\mathbf{P}(\mathbf{y}_{-\infty:k}^c)\big)\leq \log(\frac{1}{1-\epsilon^2}).
\end{equation}
\end{lemma}

  
\subsection{Formalizing Attack Goal}\label{sec:attack_goal}

As 
discussed, the attacker aims to degrade control performance by causing the position trajectory to deviate by at least $\alpha$ from the desired trajectory.  Specifically, there exists a time $k'$ after the attack starts ($k=0$) such that $\Vert \mathbf{p}_{k'}^a-\mathbf{p}_{k'}^{}\Vert \geq \alpha $ for some~$\alpha>0$. Additionally, the attacker seeks to remain stealthy (i.e., undetected by the AD), as formalized below.
\begin{definition}
\label{def:eps_alpha}
An attack sequence $\{a_{0}, a_{1},...\}$  
is referred to as $(\epsilon,\alpha)$-successful if there exists $k'\in \mathbb{Z}_{\geq 0}$ such that $ \Vert \mathbf{p}_{k'}^a-\mathbf{p}_{k'}\Vert \geq \alpha $ and the 
attack is $\epsilon$-stealthy 
for all $k\in \mathbb{Z}_{\geq 0}$.
 If such sequence exists, the system is called $(\epsilon,\alpha)$-attackable.
\end{definition}



%% file: Perfect.tex
\section{Vulnerability Analysis of UAV to GPS Attack} \label{sec:perfect}

We now derive conditions under which 
a UAV system described by~\eqref{eq:UAV_state_sp} and closed-loop dynamics~\eqref{eq:closed-loop} is vulnerable to impactful stealthy attacks, as formally defined in Section~\ref{sec:stealthy}.
%
\begin{theorem}
 The UAV system~\eqref{eq:UAV_state_sp} is $(\epsilon,\alpha)$-successful attackable for any $\alpha $ if its closed-loop dynamical model~\eqref{eq:closed-loop} is exponentially incrementally stable.     
\end{theorem}
\begin{proof}
Let us assume the trajectory of the system and controller states for $k\in \mathbb{Z}_{<0}$ is denoted by $\mathbf{X}_{-T:-1}$ according to~\eqref{eq:closed-loop}. Starting the attack at time zero, the trajectory of the system under attack, denoted by $\mathbf{X}^a$ evolves as
%
\begin{equation}\label{eq:system_attacks}
\begin{split}
{\mathbf{x}}_{k+1}^a&={\mathbf{x}}_{k}^a + \mathbf{f}(\mathbf{x}_{k}^a,\mathbf{u}_{k}^a)T_d+\mathbf{w}_k, \\
\hat{\mathbf{x}}_{k}^a &= \mathbf{g}(\hat{\mathbf{x}}_{k-1}^a,\mathbf{u}_{k-1}^a,\mathbf{y}_k^{c,a},k), \,\,
\mathpzc{X}_{k}^a=f_c(\mathpzc{X}_{k-1}^a,\hat{\mathbf{x}}_{k}^a,r_k^a),
\end{split}
\end{equation}
%
where the observation and control input are
\begin{equation}
\begin{split} 
\mathbf{y}_k^{\mathbf{p},c,a} &= \mathbf{p}_{k}^a+\mathbf{n}_k^{\mathbf{p}}+a_k,\qquad \mathbf{y}_k^{\mathbf{\Omega},c,a}= h(\mathbf{\Omega}_k^a)+\mathbf{n}_k^{\mathbf{\Omega}},\\
\mathbf{u}_k^a  &= K(\mathpzc{X}_{k}^a,\hat{\mathbf{x}}_{k}^a).
\end{split}
\end{equation}

On the other hand, if the system is not under attack during $k\in \mathbb{Z}_{\geq 0}$, the evolution of the system, sensor fusion, and controller states can be represented as $\mathbf{X}_{0:k}$. This extends the system trajectories $\mathbf{X}_{-T:-1}$ assuming no false data injection occurs. The dynamics of the drone, sensor fusion, and controller states are then described by
\begin{equation}
\begin{split}
{\mathbf{x}}_{k+1}&={\mathbf{x}}_{k} + \mathbf{f}(\mathbf{x}_{k},\mathbf{u}_{k})T_d+\mathbf{w}_k, \\
\hat{\mathbf{x}}_{k} &= \mathbf{g}(\hat{\mathbf{x}}_{k-1},\mathbf{u}_{k-1},\mathbf{y}_k^{c},k),\\
\mathpzc{X}_{k}&=f_c(\mathpzc{X}_{k-1},\hat{\mathbf{x}}_{k},r_k),\qquad \mathbf{u}_k  = K(\mathpzc{X}_{k},\hat{\mathbf{x}}_{k}),\\
\mathbf{y}_k^{\mathbf{p},c} &= \mathbf{p}_{k}+\mathbf{n}_k^{\mathbf{p}}, \qquad \,\,\,
\mathbf{y}_k^{\mathbf{\Omega},c}= h(\mathbf{\Omega}_k)+\mathbf{n}_k^{\mathbf{\Omega}}.
\end{split}
\end{equation}

In a compact form, the system dynamics are given by $\mathbf{X}_{k+1}=F(\mathbf{X}_k,\mathbf{d}_k,k)$. 
Consider the GPS attack sequence 
\begin{equation}\label{eq:attack_gps}
a_k= -Ct_k=- \begin{bmatrix}
c_x&c_y&c_z   
\end{bmatrix}^Tt_{k+1},  
\end{equation}
where $t_{k+1}=(k+1)T_d$ for all $k\in \mathbb{Z}_{\geq 0}^g$; the constants $c_x$, $c_y$ and $c_z$ in the vector $C\in \mathbb{R}^3$ are chosen by the attacker.  

Let us define $\mathbf{p}_{k}^f\delequal\mathbf{p}_k^a-Ct_{k+1}$, 
$k\in \mathbb{Z}_{\geq 0}$. Thus, the GPS measurements delivered to the controller after the attack are
\begin{equation}
\mathbf{y}_k^{\mathbf{p},c,a} = \mathbf{y}_k^{\mathbf{p},a}+a_k = \mathbf{p}_{k}^a+\mathbf{n}_k^{\mathbf{p}}-Ct_{k+1}=\mathbf{p}_{k}^f+\mathbf{n}_k^{\mathbf{p}}
\end{equation}
for all $k\in \mathbb{Z}^{g}_{\geq 0}$. 

By defining ${\mathbf{v}}_k^f\delequal{\mathbf{v}}_k^a-C$, the evolution of the $\mathbf{p}_{k}^f$ is
\begin{equation}\label{eq:pos_f}
\begin{split}
\mathbf{p}_{k+1}^f&=\mathbf{p}_{k+1}^a-Ct_{k+2} = \mathbf{p}_{k}^a+({\mathbf{v}}_k^a)T_d+\mathbf{w}_k^{\mathbf{p}}-C(k+2)T_d\\
&= \mathbf{p}_{k}^f + ({\mathbf{v}}_k^a-C)T_d +\mathbf{w}_k^{\mathbf{p}}= \mathbf{p}_{k}^f + {\mathbf{v}}_k^fT_d +\mathbf{w}_k^{\mathbf{p}},
\end{split}
\end{equation}
where $\mathbf{w}_k^{\mathbf{p}}$ is the portion of vector $\mathbf{w}_k$ associated with the position dynamics. The evolution of ${\mathbf{v}}_k^f$ can then be captured
\begin{equation}\label{eq:vel_f}
\begin{split}
\mathbf{v}_{k+1}^f&={\mathbf{v}}_{k+1}^a-C = {\mathbf{v}}_{k}^a - C +T_d\mathbf{f}_{\mathbf{v}}(\Omega^a,F_z^a)+\mathbf{w}_k^{\mathbf{v}}\\
&= {\mathbf{v}}_{k}^f +T_d\mathbf{f}_{\mathbf{v}}(\Omega^a,F_z^a)+\mathbf{w}_k^{\mathbf{v}},
\end{split}
\end{equation}
where $\mathbf{f}_{\mathbf{v}}$ represents a part of the function $\mathbf{f}$ associated with the derivative of velocity (the right-hand side of the second three equations in~\eqref{eq:UAV_state_sp}), and $\mathbf{w}_k^{\mathbf{v}}$ is the associated process noise. Also, by defining $\mathbf{\Omega}_k^f = \mathbf{\Omega}_k^a$ we obtain
%
\begin{equation}\label{eq:omega_f}
{\mathbf{\Omega}}_{k+1}^f = {\mathbf{\Omega}}_{k}^f+T_d \mathbf{f}_{{\mathbf{\Omega}}}(\Omega_k^f,\bm{\tau}^a_k),
\end{equation}
and $\mathbf{f}_{\mathbf{\Omega}}$ is the component of $\mathbf{f}$ related to the derivative of angular velocity. Combining~\eqref{eq:pos_f}-\eqref{eq:omega_f} with the dynamics of $\hat{\mathbf{x}}_{k}$ and $\mathpzc{X}_{k}^a$ from~\eqref{eq:system_attacks}, we obtain
\begin{equation}\label{eq:fake_state_d}
\begin{split}
{\mathbf{x}}_{k+1}^f&={\mathbf{x}}_{k}^f + \mathbf{f}(\mathbf{x}_{k}^f,\mathbf{u}_{k}^a)T_d+\mathbf{w}_k, \\
\hat{\mathbf{x}}_{k}^a &= \mathbf{g}(\hat{\mathbf{x}}_{k-1}^a,\mathbf{u}_{k-1}^a,\mathbf{y}_k^{c,a},k),\\
\mathpzc{X}_{k}^a&=f_c(\mathpzc{X}_{k-1}^a,\hat{\mathbf{x}}_{k}^a,r_k), \quad \,\, \mathbf{u}_k^a  = K(\mathpzc{X}_{k}^a,\hat{\mathbf{x}}_{k}^a),\\
\mathbf{y}_k^{\mathbf{p},c,a} &= \mathbf{p}_{k}^f+\mathbf{n}_k^{\mathbf{p}},\qquad \qquad
\mathbf{y}_k^{\mathbf{\Omega},c,a}= h(\mathbf{\Omega}_k^f)+\mathbf{n}_k^{\mathbf{\Omega}}
\end{split}
\end{equation}
The above system of equations can be written as
%
\begin{equation}\label{eq:closed_fake}
\mathbf{X}_{k+1}^f=F(\mathbf{X}_k^f,\mathbf{d}_k,k).
\end{equation}
with $\mathbf{X}_k^f\delequal \begin{bmatrix}
(\mathbf{x}^f)^T& ({\hat{\mathbf{x}}}^a)^T &(\mathpzc{X}^a)^T \end{bmatrix}^T$. The dynamical models~\eqref{eq:closed_fake} and~\eqref{eq:closed-loop} have the same input $\mathbf{d}$, but different initial conditions. Thus, applying the incremental stability  property of the closed-loop system yields
\vspace{-4pt}
\begin{equation}\label{eq:increm_closed}
\Vert  \mathbf{X}_{k}-\mathbf{X}_{k}^f \Vert   \leq \kappa \Vert\mathbf{X}_{0}-\mathbf{X}_{0}^f \Vert \lambda^{-k},
\end{equation}
for some $\kappa, \lambda>1$. 
From the definition $\mathbf{p}_{k}^f=\mathbf{p}_k^a-Ct_k$ and ${\mathbf{v}}_k^f={\mathbf{v}}_k^a-C$, it follows that $\mathbf{p}_{0}^f=\mathbf{p}_0^a$ and ${\mathbf{v}}_0^f={\mathbf{v}}_0^a-C$. Since it takes one time step until the attack signal impacts the system, we get ${\mathbf{x}}_{0}^a={\mathbf{x}}_{0}$. Therefore, $\mathbf{p}_{0}^f=\mathbf{p}_0$ and ${\mathbf{v}}_0^f={\mathbf{v}}_0-C$. Similarly, $\mathbf{\Omega}_0^f=\mathbf{\Omega}_0$, resulting in $\mathbf{x}_{0}-\mathbf{x}_{0}^f=\begin{bmatrix}
\mathbf{0}_9^T & c_x & c_y & c_z  
\end{bmatrix}^T$. Hence, we obtain
%
\begin{equation}\label{eq:exp_in}
\begin{split}
\Vert  \mathbf{X}_{0}-\mathbf{X}_{0}^f \Vert   \leq &\Vert  \mathbf{x}_{0}-\mathbf{x}_{0}^f \Vert + \Vert  {\hat{\mathbf{x}}_0}-{\hat{\mathbf{x}}_0^a} \Vert + \Vert  \mathpzc{X}_0-\mathpzc{X}_0^a \Vert\\ 
\leq & \Vert C \Vert + (L_{\mathbf{g}}+L_{f_c})T_d\Vert C \Vert,  
\end{split}
\end{equation}
where the above inequality uses the Lipschitz properties of the functions $\mathbf{g}$ and $f_c$. 

Using~\eqref{eq:exp_in}, we demonstrate that the attack from~\eqref{eq:attack_gps} is $\epsilon$-stealthy. From KL divergence monotonicity (Lemma~\ref{lemma:mon}), it holds that
%
\begin{equation}\label{eq:KL_mono}
\begin{split}
&KL\big(\mathbf{Q}(\mathbf{y}_{-\infty:-1}^{c},\mathbf{y}_{0:k}^{c,a})||\mathbf{P}(\mathbf{y}_{-\infty:k}^{c})\big)\leq \\
&KL\big(\mathbf{Q}(\mathbf{w}_{-\infty:k},\mathbf{y}_{-\infty:-1}^c,\mathbf{y}_{0:k}^{c,a})||\mathbf{P}(\mathbf{w}_{-\infty:k},\mathbf{y}_{-\infty:k}^{c})\big).
\end{split}
\end{equation}
%
Using the chain-rule property of KL-divergence in~\eqref{eq:KL_mono} gives
%
\begin{equation}\label{eq:chain_1}
\begin{split}
&KL\big(\mathbf{Q}(\mathbf{w}_{-\infty:k},\mathbf{y}_{-\infty:-1}^c,\mathbf{y}_{0:k}^{c,a})||\mathbf{P}(\mathbf{w}_{-\infty:k},\mathbf{y}_{-\infty:k}^{c})\big)\leq \\
&KL\big(\mathbf{Q}(\mathbf{w}_{-\infty:k},\mathbf{y}_{-\infty:-1}^{c})||\mathbf{P}(\mathbf{w}_{-\infty:k},\mathbf{y}_{-\infty:-1}^c)\big) + \\
&KL\big(\mathbf{Q}(\mathbf{y}_{0:k}^{c,a}|\mathbf{w}_{-\infty:k},\mathbf{y}_{-\infty:-1}^c)||\mathbf{P}(\mathbf{y}_{0:k}^{c}|\mathbf{w}_{-\infty:k},\mathbf{y}_{-\infty:-1}^{c})\big).
\end{split}
\end{equation}
%

The first term in the right hand side of the above inequality is zero as the  KL-divergence between two identical joint distributions is zero. Applying the chain-rule to the second term in the right hand side of 
\eqref{eq:chain_1} gives
\begin{equation}\label{eq:chain_2}
\begin{split}
&KL\big(\mathbf{Q}(\mathbf{y}_{0:k}^{c,a}|\mathbf{w}_{-\infty:k},{\mathbf{y}_{-\infty:-1}^{c}})||\mathbf{P}(\mathbf{y}_{0:k}^c|\mathbf{w}_{-\infty:k},{\mathbf{y}_{-\infty:-1}^{c}})\big) \\
&\leq \sum_{i=0}^k KL\big(\mathbf{Q}(\mathbf{y}_{i}^{c,a}|\mathbf{w}_{-\infty:k},\mathbf{y}_{-\infty:-1}^{c},\mathbf{y}_{0:i-1}^{c,a})\\
&\qquad \qquad \qquad||\mathbf{P}(\mathbf{y}_{i}^c|\mathbf{w}_{-\infty:k},\mathbf{y}_{-\infty:i-1}^{c})\big).
\end{split}
\end{equation}

It is straightforward to verify that given $\mathbf{w}_{-\infty:k},\mathbf{y}_{-\infty:i-1}^{c}$, the distribution of the $\mathbf{y}_{i}^c$ is Gaussian. For any $i\not\in \mathbb{Z}_{\geq 0}^{g}$, we have $\mathbf{y}_{i}^c=h(\mathbf{\Omega}_k)+\mathbf{n}_k^{\mathbf{\Omega}}$, which gives $\mathbf{P}(\mathbf{y}_{i}|\mathbf{w}_{-\infty:k},\mathbf{y}_{-\infty:i-1}^{c}) \thicksim \mathcal{N}(h(\mathbf{\Omega}_k), \Sigma_{\mathbf{\Omega}})$.~For~$i\in \mathbb{Z}_{\geq 0}^{g}$ we have that $\mathbf{P}(\mathbf{y}_{i}|\mathbf{w}_{-\infty:k},\mathbf{y}_{-\infty:i-1}^{c}) \thicksim \mathcal{N}(\begin{bmatrix}\mathbf{p}_k\\
h(\mathbf{\Omega}_k)\end{bmatrix}, \text{diag}(\Sigma_{\mathbf{p}},\Sigma_{\mathbf{\Omega}}))$. Similarly, given $\mathbf{w}_{-\infty:k},\mathbf{y}_{-\infty:-1}^{c},\mathbf{y}_{0:i-1}^{c,a}$   the distribution of the $\mathbf{y}_{i}^{c,a}$ is Gaussian. For any $i\not\in \mathbb{Z}_{\geq 0}^{g}$, we have that $\mathbf{y}_{i}^{c,a}=h(\mathbf{\Omega}_k^f)+\mathbf{n}_k^{\mathbf{\Omega}}$, which results in $\mathbf{P}(\mathbf{y}_{i}^{c,a}|\mathbf{w}_{-\infty:k},\mathbf{y}_{-\infty:i-1}^{c}) \thicksim \mathcal{N}(h(\mathbf{\Omega}_k^f), \Sigma_{\mathbf{\Omega}})$. For $i\in \mathbb{Z}_{\geq 0}^{g}$ we have that $\mathbf{P}(\mathbf{y}_{i}^{c,a}|\mathbf{w}_{-\infty:k},\mathbf{y}_{-\infty:i-1}^{c}) \thicksim \mathcal{N}(\begin{bmatrix}\mathbf{p}_i^f\\
h(\mathbf{\Omega}_i^f)\end{bmatrix}, \text{diag}(\Sigma_{\mathbf{p}},\Sigma_{\mathbf{\Omega}}))$. 

Using Lemma~\ref{lemma:Guassian} for $i\in \mathbb{Z}_{\geq 0}-\mathbb{Z}_{\geq 0}^{g}$, it holds that
%
\begin{equation}\label{eq:KL_imu}
\begin{split}
&KL\big(\mathbf{Q}(\mathbf{y}_{i}^{c,a}|\mathbf{w}_{-\infty:k},\mathbf{y}_{-\infty:-1}^{c},\mathbf{y}_{0:i-1}^{c,a})||\mathbf{P}(\mathbf{y}_{i}^c|\mathbf{w}_{-\infty:k},\mathbf{y}_{-\infty:i-1}^{c})\big)\\
&= \big(h(\mathbf{\Omega}_i)-h(\mathbf{\Omega}_i^f)\big)^T\Sigma_{\mathbf{\Omega}}^{-1}\big(h(\mathbf{\Omega}_i)-h(\mathbf{\Omega}_i^f)\big)\\
&\leq L_h^2\Vert \mathbf{\Omega}_i - \mathbf{\Omega}_i^f\Vert^2\lambda_{max}(\Sigma_{\mathbf{\Omega}}^{-1}).
\end{split}
\end{equation}
And for $i\in \mathbb{Z}_{\geq 0}^{g}$, we obtain
%
\begin{equation}\label{eq:KL_gps_imu}
\begin{split}
&KL\big(\mathbf{Q}(\mathbf{y}_{i}^{c,a}|\mathbf{w}_{-\infty:k},\mathbf{y}_{-\infty:-1}^{c},\mathbf{y}_{0:i-1}^{c,a})||\mathbf{P}(\mathbf{y}_{i}^c|\mathbf{w}_{-\infty:k},\mathbf{y}_{-\infty:i-1}^{c})\big)\\
&= \big(\begin{bmatrix}\mathbf{p}_i-\mathbf{p}_i^f\\
h(\mathbf{\Omega}_i)-h(\mathbf{\Omega}_i^f)\end{bmatrix}\big)^T\begin{bmatrix}
\Sigma_{\mathbf{p}}&0\\0&\Sigma_{\mathbf{\Omega}}
\end{bmatrix}^{-1}\big(\begin{bmatrix}\mathbf{p}_i-\mathbf{p}_i^f\\
h(\mathbf{\Omega}_i)-h(\mathbf{\Omega}_i^f)\end{bmatrix}\big)\\
& \leq  L_h^2\Vert \mathbf{\Omega}_i - \mathbf{\Omega}_i^f\Vert^2\lambda_{max}(\Sigma_{\mathbf{\Omega}}^{-1}) + \Vert \mathbf{p}_i - \mathbf{p}_i^f\Vert^2\lambda_{max}(\Sigma_{\mathbf{p}}^{-1}) .
\end{split}
\end{equation} 
Using the $l_2$ norm properties, we have $\Vert \mathbf{p} - \mathbf{p}^f\Vert, \Vert \mathbf{\Omega} - \mathbf{\Omega}^f\Vert \leq \Vert  \mathbf{X}-\mathbf{X}^f \Vert$ as $\mathbf{p}$ and $\mathbf{\Omega}$ are sub-vectors of $\mathbf{X}$. Combining this with inequalities~\eqref{eq:KL_mono}-\eqref{eq:KL_gps_imu} and leveraging the incremental stability of the closed-loop system~\eqref{eq:increm_closed}, we obtain that
%
\begin{equation}
\begin{split}\label{eq:b_epsilon_1}
&KL\big(\mathbf{Q}(\mathbf{y}_{-\infty:-1}^{c},\mathbf{y}_{0:k}^{c,a})||\mathbf{P}(\mathbf{y}_{-\infty:k}^{c})\big)\leq \kappa^2 \Vert \mathbf{X}_0-\mathbf{X}_0^f \Vert^2\times \\
& \Big(\sum_{i=0}^{k}L_h^2 \lambda^{-2i}\lambda_{max}(\Sigma_{\mathbf{\Omega}}^{-1})+\sum_{\substack{i=0\\ i\in \mathpzc{Z}^g}}^k\lambda^{-2i}\lambda_{max} (\Sigma_{\mathbf{p}}^{-1})\Big) \leq \\
& \kappa^2 \Vert C \Vert^2( 1+ (L_{\mathbf{g}}+L_{f_c})T_d)^2\big(
\frac{L_h^2\lambda_{max}(\Sigma_{\mathbf{\Omega}}^{-1})}{1-\lambda^{-2}}+\frac{\lambda_{max}(\Sigma_{\mathbf{p}}^{-1})}{1-\lambda^{-2N_{g}}}\big).
\end{split}
\end{equation}
Let $b_{\epsilon}$ denote the right-hand side of the inequalities. Then, we have $\epsilon=\sqrt{1-e^{b_{\epsilon}}}$. With the attack shown to be
$\epsilon$-stealthy, we now need to demonstrate its effectiveness. 

By combining $\mathbf{p}_{k}^f=\mathbf{p}_k^a-Ct_{k+1}$ with $\Vert \mathbf{p}_{k}-\mathbf{p}_{k}^f\Vert \leq \kappa \Vert\mathbf{X}_{0}-\mathbf{X}_{0}^f \Vert \lambda^{-k}$ from~\eqref{eq:increm_closed}, it holds that 
%
\begin{equation}
\begin{split}
&\Vert \mathbf{p}_{k}-\mathbf{p}_k^a+Ct_{k+1}\Vert \leq \kappa \Vert\mathbf{X}_{0}-\mathbf{X}_{0}^f \Vert \Rightarrow \\
&\Vert Ct_{k+1} \Vert - \Vert \mathbf{p}_{k}^a-\mathbf{p}_k\Vert \leq  \kappa \Vert\mathbf{X}_{0}-\mathbf{X}_{0}^f \Vert \Rightarrow \\ 
&\Vert \mathbf{p}_{k}^a-\mathbf{p}_k\Vert \geq \Vert Ct_{k+1} \Vert - \kappa \Vert C \Vert( 1+ (L_{\mathbf{g}}+L_{f_c})T_d) 
\end{split}
\end{equation}
Since $t_{k+1}= (k+1)T_d$ can become arbitrarily large as $k$ increase, we can choose $k\geq -1 + \frac{\alpha+\kappa ( 1+ \Vert C\Vert(L_{\mathbf{g}}+L_{f_c})T_d)}{\Vert C\Vert T_d}$ to ensure $\alpha$ successful attack, which concludes the proof.
\end{proof}

From~\eqref{eq:attack_gps}, the attacker does not need the system model to generate attack sequences, 
launching a stealthy and impactful attack by compromising only the GPS sensor, and not the IMU. 
With all parameters on the right-hand side of the inequality~\eqref{eq:b_epsilon_1} being constants, stealthiness can be controlled by choosing smaller values for 
$\Vert C\Vert $, though this will prolong the time required for the attack to be $\alpha$-effective. Thus, there is a trade-off between stealthiness and effectiveness time.
The GPS spoofed signals, 
$\mathbf{y}_{k}^{c,a}$, follow the dynamics~\eqref{eq:fake_state_d}. As $\mathbf{X}_{k}^f$  closely matches the UAV trajectory without the attack, the spoofed GPS signals appear legitimate to~the~controller. 

%% file: Simulation.tex
\section{Simulation Results}
\label{sec:simulation}

\begin{figure*}[!t]
\begin{subfigure}{0.32\textwidth}
\includegraphics[width=.9\linewidth, height=3.8cm]{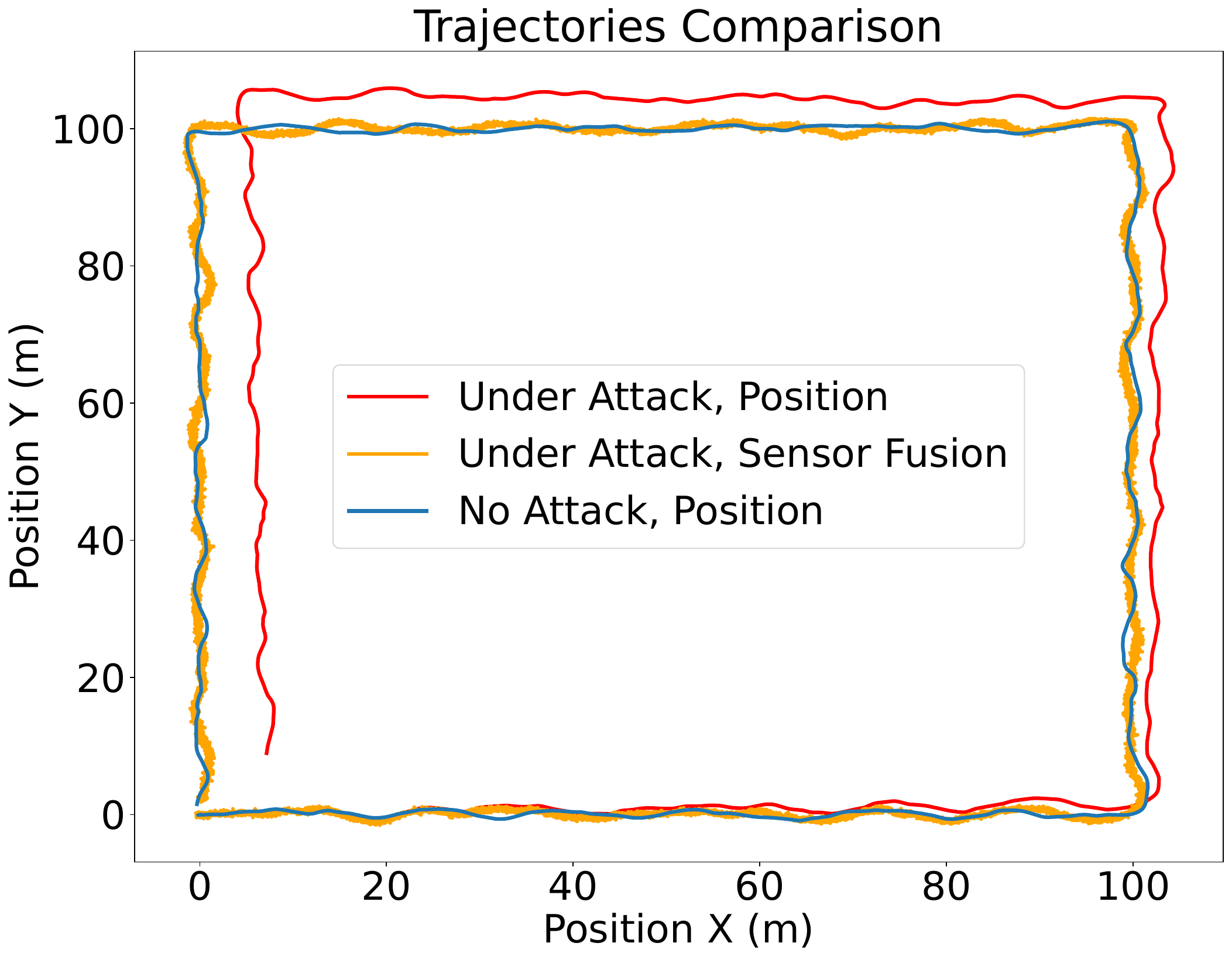}
\hspace{-6pt}
\label{fig:traj}
\end{subfigure}
\begin{subfigure}{0.32\textwidth}
\includegraphics[width=.9\linewidth, height=3.8cm]{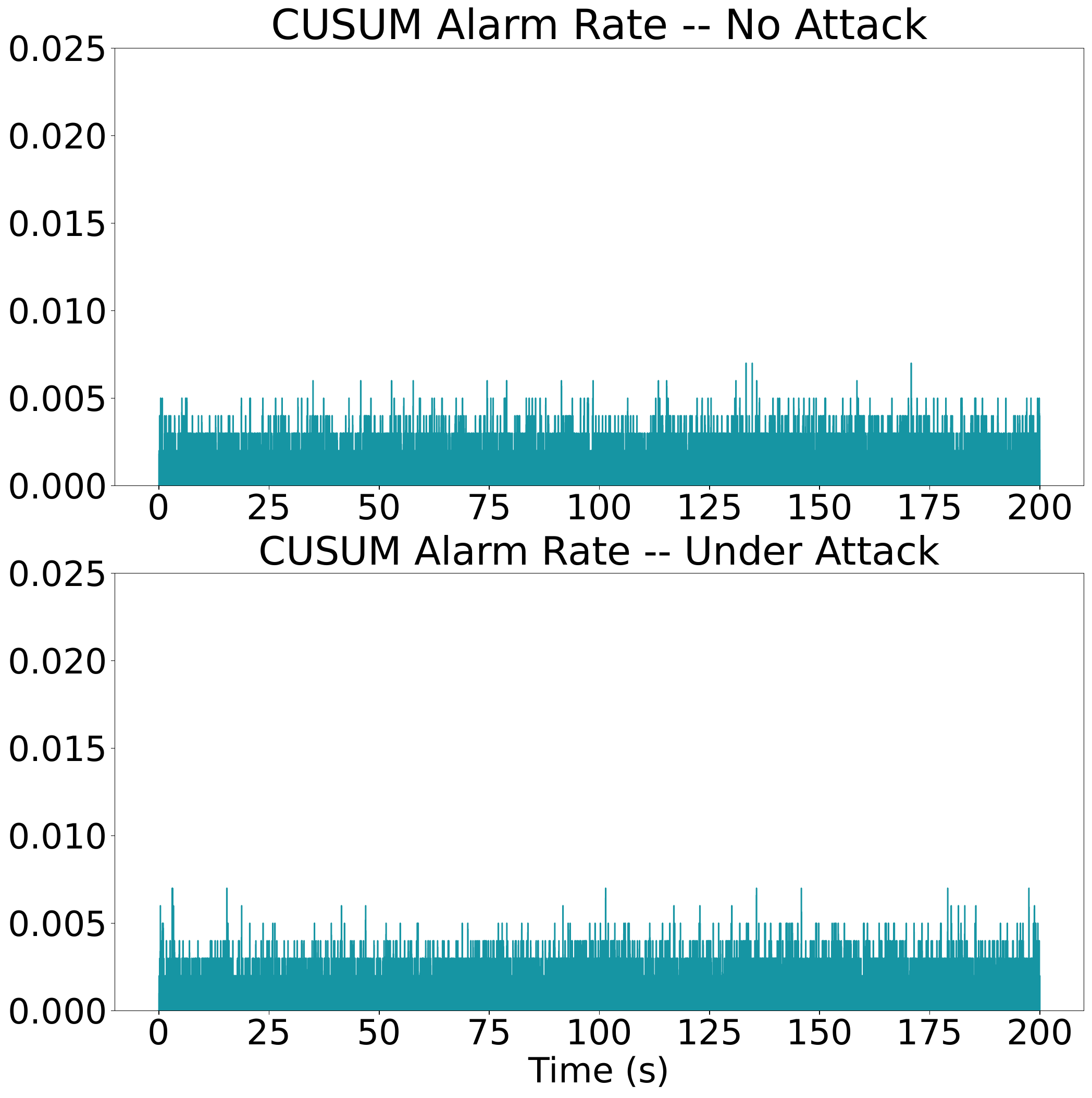}
\hspace{-6pt}
\label{fig:drone_CUSUM}
\end{subfigure}
\begin{subfigure}{0.32\textwidth}
\includegraphics[width=.9\linewidth, height=3.8cm]{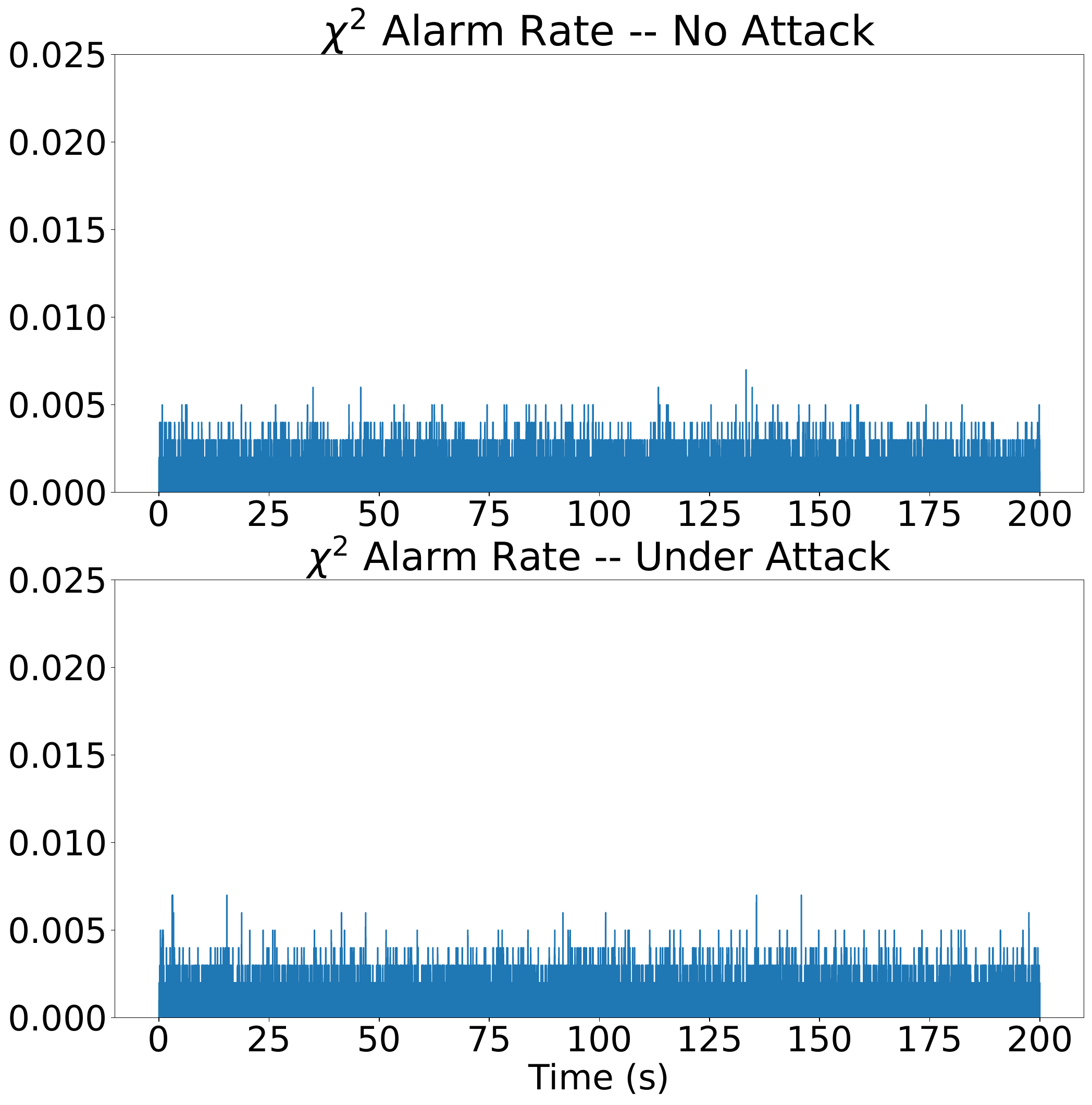}
\hspace{-6pt}
\label{fig:drone_Chi}
\end{subfigure}
\caption{(a)~The $X-Y$ position trajectories of the drone without the attack (blue), with the attack (red) and the output of sensor fusion for $X-Y$ position in the presence of attack. (b,c)~The alarm rate for CUSUM and $\chi^2$ ADs  for two cases when the system is under attack (attack starts at time t = 0) and when it is attack free over 5000 experiments.}
\label{fig:drone_stealthiness_fig}
\vspace{-4pt}
\end{figure*}

We modeled quadcopter dynamics with disturbances such as drag forces on the torque and thrust, generating IMU and GPS signals with Matlab's built-in functions. For sensor fusion, we use Matlab's Navigation toolbox and a cascade PID controller for attitude and position control. The position controller's output serves as the set-point for the attitude controller. We also implement two ADs: $\chi^2$ and CUSUM~\cite{umsonst2017security}, with a false alarm threshold of $p^{FA}=0.001$.

We simulate a waypoint tracking task, where the drone navigates through four points forming a square in the $X-Y$ plane: starting at $(0, 0, 10\text{m})$, moving to $(100\text{m}, 0, 20\text{m})$, $(100\text{m}, 100\text{m}, 20\text{m})$, $(0, 100\text{m}, 10\text{m})$, and returning to $(0, 0, 10\text{m})$. Using the attacks from~\eqref{eq:attack_gps} with $\Vert C\Vert = 0.05$, Fig.~\ref{fig:drone_stealthiness_fig}a 
shows the drone's trajectory without attack (blue) and with attack (red). The attack causes the drone to deviate by $10m$ from its intended endpoint. The orange line represents the sensor fusion output, which aligns with the trajectory in the absence of attack, indicating the system would perceive normal operation. Fig.~\ref{fig:drone_stealthiness_fig}b 
and  Fig.~\ref{fig:drone_stealthiness_fig}c 
show the alarm rate of  CUSUM and $\chi^2$ ADs for both cases when the system is free of attack and when the GPS is compromised (the attack starts at time zero) -- 
the average true detection rate closely matches the average false detection rate for both ADs, demonstrating the attack's stealthiness.





%% file: Conclusion.tex
\section{Conclusion}
\label{sec:concl}
In this work, we analyzed the vulnerability of UAV 
to black-box stealthy false data injections on GPS sensors. We defined stealthiness as an attack that remains undetected by any anomaly detector. We demonstrated that if the system is closed-loop incrementally stable, by only compromising the GPS sensors, 
stealthy attacks can cause arbitrarily large deviations in the drone's position trajectories. Moreover, we established that the attacker does not need access to the system model or any state information to execute~the~attacks. 